%% file: paper.tex
\documentclass[11pt]{article}
\usepackage{geometry}[1in]
\pdfoutput=1

\usepackage[fleqn]{amsmath}
\usepackage{amssymb, amsthm} 
\newtheorem{theorem}{Theorem}[section]

\newtheorem{lemma}[theorem]{Lemma}

\newtheorem{proposition}[theorem]{Proposition}

\newtheorem{defn}[theorem]{Definition}
\newtheorem{definition}[theorem]{Definition}

\usepackage{fullpage, graphics}

\usepackage{hyperref} 

\hypersetup{colorlinks,
 citecolor=MidnightBlue,
  linkcolor=black,
  urlcolor=black}
\usepackage{mathpazo} 
\usepackage{amsfonts}
\usepackage{amssymb,amsmath}
\usepackage{latexsym}
\usepackage{graphicx}
\usepackage{color}
\usepackage[dvipsnames]{xcolor}
\usepackage{enumerate}

\newcommand{\ea}{{\em et al.}\ }

\def\E{\mathop{\mathbb{E}}\limits}
\newcommand{\Ebb}{\mathbb{E}}
\def\Pr{\mathop{\mathbf{Pr}}\limits}

\newcommand{\deq}{\stackrel{\scriptscriptstyle\triangle}{=}}

\renewcommand{\epsilon}{\varepsilon}

\newcommand{\dist}{\mathsf{dist}}

\newcommand{\R}{{\mathbb R}}
\newcommand{\Z}{{\mathbb Z}}

\newcommand{\C}{\mathcal{C}}

\renewcommand{\a}{\alpha}

\renewcommand{\d}{\delta}
\newcommand{\g}{\gamma}

\newcommand{\bits}{\{0,1\}^d }

\newcommand{\A}{\mathcal{A}}
\newcommand{\Ak}{\mathcal{A}_k(\delta)}

\newcommand{\set}[1]{\ensuremath{\{#1\}}}

\newcommand{\Dn}{\mathcal{D}_{n,d}}

\newcommand{\Bk}{\Ball( 0^d, k)}

\newcommand{\edges}{\mathsf{edges}}

\newcommand{\er}{\mathsf{e}_{\leq r}}
\newcommand{\Er}{\mathsf{E}_{\leq r}}
\newcommand{\eR}{\mathsf{e}_{\leq R}}
\newcommand{\ER}{\mathsf{E}_{\leq R}}

\newcommand{\edr}{\mathsf{e}_{r}}
\newcommand{\Edr}{\mathsf{E}_{r}}

\newcommand{\G}{\Gamma_{\leq r}}

\newcommand{\Gdr}{\Gamma_{r}}

\newcommand{\Ball}{\mathsf{Ball}}

\newcommand{\ceil}[1]{\lceil #1 \rceil}
\newcommand{\floor}[1]{\lfloor #1 \rfloor}

\newcommand{\lk}{\ell_k}

\newcommand{\fixb}{
  Fix $b\in \{0,1,\ldots,\floor{r/2}$\}, 
  and assume that $r$ divides $R$ with $r < R$
 }
\newcommand{\denom}{4^{R/r} \cdot R! \cdot d^{bR/r}}

 \newcommand{\ES}{\E_{S\sim \Dn}}

\renewcommand{\o}{\mathsf{overhead}}

\newcommand{\w}{\mathsf{w}}
 
\newcommand{\ohbound}{\ensuremath{f_{\le r}(n,d,p,\delta)}}
\newcommand{\eqohbound}{\ensuremath{f_r(n,d,p,\delta)}}
\newcommand{\numberthis}{\addtocounter{equation}{1}\tag{\theequation}}
\newcommand{\ignore}[1]{\relax}

\usepackage{fancyhdr}
\title{Massively-Parallel Similarity Join, Edge-Isoperimetry, and Distance
Correlations on the Hypercube
} 

\author{
	Paul Beame\thanks{Research supported in part by NSF grants CCF-1217099 and CCF-1524246}
	\\ {\small Dept. of Computer Science \& Engineering}
	\\ {\small University of Washington, Seattle}
	\\ {\small beame@cs.washington.edu}
	\and 
	Cyrus Rashtchian$^*$
	\\ {\small Dept. of Computer Science \& Engineering}
	\\ {\small University of Washington, Seattle}
	\\ {\small cyrash@cs.washington.edu}
}

\begin{document}

\date{\today}
\maketitle

\input{0-abstract}

\input{1-intro}

\input{2-notation}
\input{3-results}

\input{4-edge-iso-upper}

\input{5-lower-bound}

\input{6-conclusion}
\input{7-acknowl}

\bibliographystyle{plain}
\bibliography{mybib}

\appendix
\newpage
\input{A-appendix}
 \end{document}

%% file: 0-abstract.tex
\begin{abstract}
We study distributed protocols for finding all pairs of similar vectors in a large dataset.  Our results pertain to a variety of discrete metrics, and we give concrete instantiations for Hamming distance.  
In particular, we give improved upper bounds on the overhead required for similarity defined by Hamming distance $r>1$ and prove a lower bound showing qualitative optimality of the overhead required for similarity over any Hamming distance $r$. 
Our main conceptual contribution is a connection between similarity search algorithms and certain graph-theoretic quantities.  
For our upper bounds, we exhibit a general method for designing one-round protocols using edge-isoperimetric shapes in similarity graphs.   For our lower bounds, we define a new combinatorial optimization problem, which can be stated in purely graph-theoretic terms yet also captures the core of the analysis in previous theoretical work on distributed similarity joins.  As one of our main technical results, we prove new bounds on distance correlations in subsets of the
Hamming cube. 
\end{abstract}

%% file: 1-intro.tex
\section{Introduction}
The task of reporting all pairs of similar items arises as a core computation in many applications.  Known as a similarity join in the database community, this task generalizes a relational join~\cite{afrati-fuzzy, jiang, 
vernica}.   Similarity joins aid in the removal of near-duplicate data~\cite{broder,  xiao, yalniz} and the search for images with similar content~\cite{iqt,li, wang-learning}.  In social media and other web applications, collaborative filtering methods employ similarity joins to find related users and items~\cite{googlenews, wtf, mmds, zhang-collab}.  

We give new parallel algorithms and prove new communication lower bounds for computing a similarity join in a distributed, shared-nothing cluster.  The similarity join problem is defined relative to a metric space $(V,\rho)$.  Given an input set $S$ of points in $V$ and a distance threshold~$r$, the goal is to output all pairs $x,y$ in $S$ satisfying $\rho(x,y) \leq r$.  For the parallel version of this problem, the input $S$ starts equally partitioned among $p$ processors.  
The challenge in designing an efficient parallel algorithm comes from outputting every close pair while balancing the workload among the processors and completing the overall similarity join quickly. 

Our algorithms and lower bounds are for a {\em simultaneous} model, where the processors have one round to communicate input points before locally computing and outputting their portion of the set of similar pairs.  This model is consistent with previous work~\cite{sharir, bks, afrati-dist, disco}, and with the fact that  each additional round of communication greatly increases the runtime in practice. 
We adopt the natural assumption that the local computation time
increases proportionally to the maximum number of received vertices, and we focus
on minimizing the latter quantity. 

If there are $p$ processors and $n$ input points, an ideally balanced load 
with no replication of data would send precisely $n/p$ points to each processor.
However, since every similar {\em pair} must end up at some processor, 
some replication is required, and it will not be possible to achieve this ideal. We measure the {\em overhead} of an algorithm as the ratio of the
maximum number of points any processor receives to the $n/p$ ideal\footnote{Overhead resembles  the replication rate used in analyzing
MapReduce algorithms, but it does not depend on the notion of reducer size~\cite{afrati-dist}, and thus it applies more generally than to just MapReduce algorithms.}.
To avoid large overhead in the worst case, any adequate algorithm must employ a randomized load balancing strategy. 

We exhibit new algorithms that improve on the overhead of the
best previous similarity join algorithms from~\cite{afrati-fuzzy, afrati-anchor, afrati-dist}  for Hamming distance
on $\bits$ with distance thresholds $r>1$.  Our improvements are both quantitative in terms of the overhead bound and more general in terms of the range of set sizes $n$ for which they work. More fundamentally, we show that every randomized similarity join algorithm for
Hamming distance $r$ on $\bits$, even one that only weakly approximates the
similarity join, requires overhead $d^{\Omega(r)}$.  This qualitatively matches
the overhead of the best algorithms. Moreover, 
ours is the first lower bound for similarity join that applies for any
$r\ne 1$, for any set size that is $o(2^d)$, or for any algorithm that is only
approximately correct.

We develop our improved similarity join algorithm as well as our lower bound by
analyzing certain graph-theoretic quantities of a similarity graph. For
distance threshold $r$ on $\bits$, the Hamming similarity  graph has edges between all pairs of
points with Hamming distance at most $r$.  To approximate the
similarity join, both endpoints of most edges of this graph must arrive together at some processor.  This criterion has a natural interpretation in graph terms.  Let $A_1,\ldots, A_p$ be the sets of points that are sent to each of the $p$ processors.  Then, the collection of~$p$ subgraphs induced by these sets must contain many of the edges in the input set.  Since the input consists of an arbitrary set of points, it follows
that randomized algorithms correspond to distributions $\A$ on approximate coverings $(A_1,\ldots, A_p)$ of the edges of
the similarity graph.  We refer to such a $p$-tuple as an {\em edge-covering}.

The design of edge-coverings connects the similarity join problem with the edge-isoperimetric question in graph theory.  To see this connection, notice two facts.  First, our overhead measure is minimized by reducing the number of input points contained by each set $A_i$.  Second, an edge-covering is improved if each set $A_i$ contains many edges relative to its size.  {\em Edge-isoperimetric sets} in a graph are sets of vertices of a given size that
maximize the number of induced edges~\cite{bezrukov}.  Therefore, we observe that it is good to build each set $A_i$
as a union of randomly chosen edge-isoperimetric sets of suitable size
in the similarity graph. In particular, we choose that size to be roughly $n/p$, the
ideal number of input elements sent to a processor. This ensures bounded overhead and avoids imbalanced correlations of the input
set with the edge-isoperimetric sets.  In fact, this class of algorithms is
a generic one that is suitable for use with any edge-transitive similarity
graph, such as the graphs defined by the Euclidean or Manhattan distance.
 
In the case of Hamming distance $r=1$ on $\bits$, the optimal
edge-isoperimetric sets are near-subcubes~\cite{bernstein, harper, hart, lindsey}.  It is 
worth noting that random subcubes are a feature of the best similarity join
algorithm for this case~\cite{afrati-fuzzy, afrati-dist}.  However, for $r>1$, previous algorithms
use other methods~\cite{afrati-fuzzy, afrati-anchor, HIM, coverlsh1,  afrati-dist, xiao}, and good edge-isoperimetric sets were not
identified as suitable constructs.  Though the optimal edge-isoperimetric
sets for $r>1$ are not known, in our improved similarity join algorithm we use
Hamming balls (which interestingly are vertex-isoperimetric sets for $r=1$)
and show that they are good approximations to optimal edge-isoperimetric sets.

Our general lower bound applies much more broadly than this edge-isoperimetric set
construction.  It holds for any randomized way of choosing sets
$(A_1,\ldots,A_p)$.  We show that edge-covering and similarity join are hard
even if the input set of points $S$ is itself a random, sub-sampled
Hamming ball.

The general intuition for our lower bound proof is that either the total number of elements
among the $A_i$ is large, in which case large overhead follows immediately, or
the sets $A_i$ have a large density of pairs of points in $\bits$ of distance
$\le r$.  For this second case we derive our lower bound on the overhead 
by proving a new connection between the density of different
Hamming distances of subsets on $\bits$.  
We show that for $A\subseteq\bits$ that if there is a sufficiently
high density of pairs at Hamming distance $r$ in $A$, then there is also a high
density of pairs at much larger Hamming distances.
This property is likely to be of independent interest since it can be viewed as a Sidorenko-type result~\cite{conlon, erdos-sim, sidorenko} for shortest paths in the $r$-th power of the hypercube.

\subsection*{Related Work} 

We review two baseline algorithms for the similarity join problem under any metric.
The first is a randomized join algorithm and has overhead $O(\sqrt{p})$.  We discuss the details of this universal all-pairs algorithm in Section 2.  The second baseline (known as the ``Ball-hashing-2'' algorithm in~\cite{afrati-fuzzy}) works by randomly partitioning the points $z$ of the metric space among the
processors and sending each input point $x$ to the processors associated
with all $z$ at distance at most $\ceil{r/2}$ from~$x$.  In other words, processors are responsible for unions of balls of radius $\ceil{r/2}$.  
The expected overhead of this algorithm is the size of balls of radius
$\ceil{r/2}$ in the metric space, which in the case of the Hamming distance
on $\bits$ we denote by $B(d,\ceil{r/2})$ and is 
$O(d^{\ceil{r/2}})$.

Recent work~\cite{afrati-fuzzy,afrati-anchor, afrati-dist} demonstrates improvements over the baseline algorithms for $\bits$ with Hamming distance.  For distance $r = 1$, Afrati~\ea~\cite{afrati-fuzzy} present an algorithm (called ``Splitting'') based on a fixed grid of subcubes. They analyze how subcube dimensionality affects the overall communication. We observe that choosing the dimension to be $\floor{\log_2 (n/p)}$ leads to an edge-covering  that achieves overhead $O(d/\log(n/p))$.  For $r >1$, in~\cite{afrati-fuzzy,afrati-dist} they also provide a variety of related algorithms.  They analyze subcubes for larger distances, but this is not as effective as the Ball-hashing-2 algorithm in terms of overall communication.  Their Ball-hashing-1 algorithm sends input points to processors based on balls of radius $r$ centered at input points.  However,  the Ball-hashing-1 communication scales with $B(d,r)$, worse than the Ball-hashing-2 algorithm.

The best previous algorithm for $r>1$ is the {\sc AnchorPoints} from~\cite{afrati-anchor}, which improves upon~\cite{afrati-fuzzy}.  The basic building block of their algorithm is a \emph{covering code} $\C$ of radius $r$, defined as a set of points $\C$ in $\bits$ with the property that every $x\in \bits$ is within distance $r$ of some element of $\C$.  After randomly partitioning the elements $z$ of the code $\C$ among
the processors, the algorithm sends each input point $x$ to every processor 
associated with an element of the covering code that is at distance at most
$\ceil{3r/2}$ from $x$.   In other words, processors are responsible for unions of balls of radius $\ceil{3r/2}$.  By using covering codes of density at most
$O(1/B(d,r))$, the work in \cite{afrati-anchor} improves on the overhead of the Ball-hashing-2
algorithm by roughly a $2^{\ceil{r/2}}$ factor.  
By basing our construction on nearly edge-isoperimetric sets and larger Hamming balls, 
our algorithm improves the previous bound by factor of roughly 
$(\frac{1}{2}\log_d (n/p))^{\ceil{r/2}}$.

Similarity search algorithms often employ Locality Sensitive Hashing (LSH)~\cite{charikar, HIM, kor}.  An LSH family is any distribution over hash functions such that near points map to the same bucket with a higher probability than far points.  For Hamming space, an example LSH family picks a subset of $k$ coordinates in $\{1,2,\ldots, d\}$ and maps a vector to its $k$-bit projection on these coordinates~\cite{HIM}.  LSH may be used as part of a similarity join algorithm by simply hashing the input points, checking all pairs of points within the same hash bucket, and outputting the close pairs.  The success and approximation of such an algorithm depends on $k$ and is amplified through independent trials. 

The randomness in classical LSH schemes~\cite{charikar, HIM, kor} leads to the
possibility of both false positives and false negatives, but recent work
remedies things for false negatives~\cite{coverlsh1, coverlsh2}. 
Their LSH family for Hamming distance covers all close pairs with a correlated family of subcubes. 

In a MapReduce setting, LSH buckets correspond to groups of points sent together.  Hashing a vector in $\bits$ by choosing $k$ coordinates corresponds to partitioning $\bits$ into subcubes of dimension $k$ and determining the subcube to which the vector belongs. 
Achieving a balanced load requires that each hash bucket contain a bounded
number of vectors. 
In contrast, all-pairs similarity search algorithms based on LSH ignore the
sizes of each bucket.  Instead they focus on the total work done by the
algorithm, and, over many rounds, maintain a counter for each vector and simply
abort after comparing it to too many far vectors~\cite{HIM, pagh:ioefficient}. 
In fact, as we show, the buckets maintained by LSH algorithms for distances
$r>1$ do not lead to the best load balance.


The approximate similarity join problem we study requires a sharp distance
threshold -- only producing pairs with distance at most $r$; false positives
can be easily filtered before being output. 
Therefore, our notion of approximation only refers to
the fraction of false negatives.  
This differs from problems such as Approximate Near Neighbor (ANN) search in
terms of the notion of approximation~\cite{harpeledbook}, where 
ANN benefits from allowing false positive pairs with distance $\le cr$, where
$c > 1$ determines the overall space usage and running time.  

There are well-known lower bounds for LSH, such as~\cite{ar-lb, motwaniLSH,odonnellLSH}, but these only apply for $r = \Omega(d)$. 
Indeed, they employ analytic techniques, such as the Bonami-Beckner inequality,
that do not capture the small distances that are the focus of our work.  

Panigrahy, Talwar, and Wieder~\cite{panigrahy:ann2008,panigrahy:ann2010} prove cell-probe lower bounds for randomized algorithms solving the approximate near neighbor search problem (recent improvements appear in~\cite{andoni:cellprobe}).   Their approach resembles ours in that they consider (a generalization of) the edge expansion of the similarity graph for a metric space (i.e., the graph with edges connecting every pair of points with distance at most $r$), though the 
specific application and techniques are different.

Recently, there has been work on proving strong conditional lower
bounds (assuming the strong exponential time hypothesis~\cite{ip:exptimehyp})
on the time complexity of computing all-pairs (approximate) nearest
neighbors~\cite{ahle:complexity-simjoin, alman2016, alman2015}.  However
these bounds do not apply to thresholded similarity join or to our model,
which measures the communication required for a single round.

There are also other kinds of algorithms for similarity search. 
In particular, there is a substantial body of work on
the design and analysis of data structures that first preprocess the data and 
then answer nearest-neighbor and range queries. 
Most of the analysis considers time complexity and space usage for serial
algorithms.  Therefore, the majority of these results are orthogonal to the
study of distributed similarity joins and edge coverings (cf. \cite{harpeledbook}).

For metrics other than Hamming distance, objects resembling edge-coverings
underlie state-of-the-art
algorithms for similarity joins on datasets of real vectors:
Working in Euclidean space,  Aiger, Kaplan, and Sharir~\cite{sharir} demonstrate an efficient edge-covering using randomly shifted and rotated grids.  
In the $\ell_\infty$ metric, Lenhof and Smid~\cite{lenhof} provide serial and shared-memory algorithms using a data-dependent partition of space into $\ell_\infty$ hypercubes. 
For angular distance, LSH schemes use random half-spaces or spherical caps as randomized edge-coverings~\cite{falconn, charikar, HIM, plsh, lsh-practical}. 

Many authors discovered the optimal edge-isoperimetric shape for the standard (distance
one) hypercube~\cite{bernstein, harper, hart, lindsey}.
Bezrukov's survey~\cite{bezrukov} contains a thorough, modern treatment.
Kahn, Kalai, and Linial~\cite{kkl} consider the generalization for $r$ greater
than one, but leave asymptotic results for all set sizes as a open question. 
Bollob{\'{a}}s and Leader~\cite{BollobasL96} prove a generalization of the
hypercube result for both Hamming and $\ell_1$ distance on vectors over a
larger alphabet (only for distance one in both cases).  
Our  edge-covering method implies that the Bollob{\'{a}}s and Leader result
can be used to design protocols for this space.  

Finally, in the combinatorics community, researchers have studied optimal coverings for pairs of elements in designs~\cite{chee, horsley}.  This relates to theoretical algorithms for parallel joins~\cite{bks}.  We exploit the additional structure of the metric space to design vastly more efficient protocols and algorithms than those for all-pairs joins when the number of processors is very large.

%% file: 2-notation.tex
\section{Similarity Graphs, Edge-Coverings, and Overhead}


\paragraph{Similarity Graphs.} Any undirected graph $G=(V,E)$ can be viewed as defining a notion of similarity
on $V$ in which similar pairs are joined by an edge.  We study
\emph{edge-transitive} graphs, those having automorphisms that make
every neighborhood look like any other.  We particularly focus on the 
case that $V$ is the set $\bits$ equipped with the Hamming distance
$\dist(u,v)$ that equals the number of bits differing in $u,v \in \bits$.  
For a distance threshold $r \in [d]$, 
let $\G$ denote the Hamming similarity graph with vertices $\bits$ and edges 
connecting $u,v$ whenever $\dist(u,v)\leq r$, with self-loop edges for $u = v$. 
This regular graph has degree $B(d,r) \deq \binom{d}{r} + \binom{d}{r-1} + \ldots + d + 1$.
For a subset $A$ in $\bits$, let $\Er(A)$ denote the set of edges in $\G$ 
with both endpoints in $A$.  Similarly, define $\Gdr$ as the graph with edges connecting vertices with distance 
exactly $r$, and define $\Edr(A)$ as the set of edges in $\Gdr$ with both endpoints in $A$.  
It will be convenient to define notation for the average density of edges in the induced subgraphs
$\G[A]$ and $\Gdr[A]$, so we define $\er(A) \deq \frac{|\Er(A)|}{|A|}$
	and 
	$\edr(A) \deq \frac{|\Edr(A)|}{|A|}$. \vspace{.1in}

\subsection*{Similarity Join Algorithms and Overhead}
We study a broad, natural class of one-round, $p$-processor similarity join
algorithms.  
Each point $x$ in the input set $S$ is sent to some non-empty subset of processors $P(x)\subseteq [p]$.  The algorithm then uses local operations at each
destination processor to report all similar pairs among the elements it
receives from $S$.  
We allow the mapping $P$ to depend on the (approximate) size $n$ of
the set $S$ but not $S$ itself.  
The goal is to minimize the maximum number of elements
received by any processor, while ensuring that (almost) every similar pair
in $S$ goes to at least one processor.  

The ideal load balancing would be a load of $n/p$ elements per processor,
but we show that this is not always possible.
Randomization is essential in defining $P$, since for any deterministic mapping
$P$ any set $S\subseteq P^{-1}(i)$ for any $i\in [p]$ will yield a load of
$|S|$.
It is convenient to use an equivalent and convenient perspective by considering
the distribution $(P^{-1}(1),\ldots,P^{-1}(p))$ of tuples of sets produced by
the algorithm.  This emphasizes the combinatorial aspects of the problem, and
is formalized in the following definition in terms of randomized edge-coverings.

\begin{definition}
A \emph{randomized $(p,\delta)$-edge covering} for the $n$-subsets of $\G$
is a distribution $\A$ on $p$-tuples
$(A_1,\ldots,A_p)$ of subsets of $\bits$ such that for every subset
$S\subseteq\bits$ of size $n$,
\begin{equation*}
\E_{(A_1,\ldots,A_p)\sim \A}\left[|\bigcup_{i} \Er(A_i\cap S)|\right]\ge
\delta\cdot |\Er(S)|.
\end{equation*}
There is an analogous definition of a $(p,\delta)$-edge covering for $n$-subsets
of $\Gdr$ with $\Edr$ replacing $\Er$.
\end{definition}

As a way to capture the maximum load of randomized algorithms, our primary complexity measure will be the  expectation of the maximum ratio of $|A_i\cap S|$ to $|S|/p$.  

\begin{definition}
Let $\A$ be a distribution on $p$-tuples $(A_1,\ldots,A_p)$ of subsets of
$\bits$.
Let $S\subseteq \bits$ with $p \leq |S|$.  
Define the \emph{overhead} of $\A$ on set $S$,
\begin{equation*}
	\mathsf{overhead}(\A,S) =  \E_{(A_1,\ldots,A_p)\sim \A}
	\left[ \ \max_{i \in [p]} |A_i \cap S|\cdot \frac{p}{|S|}\right ],
\end{equation*}
where the expectation is over the randomness of $\A$.  The $n$-\emph{overhead} of $\A$ is defined as
\begin{equation*}
	\o_n(\A) = \max_{S\subseteq \bits,\ |S|=n} \o(\A,S)
\end{equation*}
\end{definition}

Note that defining the overhead in terms
of the {\em maximum} value of $|A_i\cap S|$ is critical, since a
trivial algorithm that sends $S$ entirely to a random processor achieves
average load exactly $|S|/p$. 

Any randomized $(p,\delta)$-edge covering $\A$ 
yields a natural one-round, $p$-processor, data-parallel algorithm 
for computing a $\delta$-approximate similarity join.   First, choose
$(A_1,\ldots,A_p)\sim\A$ using shared randomness.  Then, send $x$ 
to processor $i$ whenever $x\in A_i$.  In other words, on input $S$, processor $i$ receives $A_i \cap S$ for $i \in [p]$.
Over the random choices of the algorithm, this strategy outputs in expectation at least a
$\delta$ fraction of all pairs in $S$ with Hamming distance at most $r$.

These 
\emph{local} algorithms capture a very
broad class of natural algorithms for similarity joins including the ones
discussed in the introduction.
For example, in the ball-hashing-2 algorithm, 
the sets $A_1,\ldots,A_p$ are unions over a random partition 
of Hamming balls of radius~$\lceil r/2 \rceil$.   

\paragraph{A Universal All-pairs Algorithm.}
As another example of such algorithms we present the other baseline
load-balancing algorithm mentioned in the introduction which works for any
notion of similarity and yields overhead $O(\sqrt{p})$.  

\begin{proposition}
There is a randomized 1-round local algorithm  for any similarity measure that
has expected per-processor load $O(|S|/\sqrt{p})$ on all sets $S$ with
$p\le |S|$.
\end{proposition}

\begin{proof}
Assume without loss of generality that $p=\binom{q}{2}$ for integer $q$.
Choose a random mapping $h:\bits\rightarrow [q]$.  Now send $x\in \bits$
to processors indexed by the pair $\set{h(x),i}$ for every $i\in [q]$ not
equal to $h(x)$.  Every pair of inputs in $S$ will be seen by some processor.
The per-processor load is almost surely $O(|S|/q)$.
\end{proof}

This yields a randomized $(p,1)$-edge-covering of $\bits$
 with overhead $O(\sqrt p)$ for all set sizes.
In Appendix~\ref{otherbounds} we note that for very small sets $S$, this bound
essentially cannot be improved, even on $\G$.  
For a limited number of available processors, this algorithm is 
still useful in practice, even for large datasets.
However, many important cases involve sets $S$ that are a much larger
portion of $\bits$.  In fact, some previous work on similarity join on
the binary cube has often focused on the case that the input set size is
completely at the other extreme -- a constant fraction of all possible vectors.
We consider a wide range of input sizes and not just 
these extremes.


\begin{definition}[\textbf{Main Complexity Measure}]
Let $\ohbound$ be the minimum of $\o_n(\A)$
over all randomized
$(p,\delta)$-edge coverings for  $n$-subsets of $\G$.  Analogously define $\eqohbound$ for $\Gdr$. 
\end{definition}

Lower bounds on $\ohbound$ are lower bounds on the amount of
overhead required for any randomized $p$-processor 1-round local algorithm
that approximately solves similarity join on the binary cube with Hamming
distance threshold $r$ for subsets $S\subseteq \bits$ of size $n$.
Upper bounds on $\ohbound$ and $\eqohbound$
will be particularly interesting when $\delta=1$, which implies no error, or
when $\delta$ is sufficiently close to 1, which implies the probability of missing any
similar pair in $S$ is extremely small.

\paragraph{Comparison to MapReduce Measures.} 
In MapReduce algorithms, there is a collection of individual tasks called
\emph{reducers} that are randomly assigned to one of the $p$ processors.  
Previous similarity join algorithms (as well as our edge-isoperimetric
algorithm) define reducers for each subset  $R_1,\ldots,R_{K}$ of
$\bits$ for some $K \gg p$.  Then, during a shuffling phase, the sets
$A_1,\ldots, A_p$ are formed randomly as unions of some of the $\{R_j\}$.   
The key complexity measure analyzed in the previous
work~\cite{afrati-fuzzy,  afrati-anchor, afrati-dist} is the
\emph{replication rate}, defined as the average number of reducers to which
each point is sent.  We claim that overhead is a better measure.

The replication rate is a useful measure only when the reducers are small
enough that they can be mapped randomly and uniformly to yield tasks that are nicely load-balanced across the processors. Our measure, overhead, captures a much wider class of algorithms than replication rate can. Moreover,  in the regime where  replication rate is a useful measure, overhead essentially equals replication rate.  The main advantage of overhead concerns very large reducers.  For example, if one reducer was responsible for essentially the whole input, the replication rate would be one, but this algorithm would behave badly in
practice and has the worst possible overhead of $p$. Therefore, overhead both captures the lower bounds using replication rate and defines a measure of load balancing for large reducers. 
%

%% file: 3-results.tex
\section{Our Results}\label{sec:results}

\begin{theorem}\label{thm}\label{mainlb} 
Let $0<\delta\le 1$ and $r,d,n,p$ be such that $r\le d$, $p \leq n\le 2^d$, and
$\g = \log_n p$.

\begin{enumerate}[(a)]
\item For all integers $k$ with $\ceil{r/2}\le k\le (1-\gamma)\log_d n$
there is a MapReduce algorithm witnessing the bound
\begin{align*}
&\ohbound
\leq \max\left\{6 \cdot 2^{\ceil{r/2}}\ln\left(\frac1{1-\delta}\right) \cdot  \left(\frac{d}{k}\right)^{\ceil{r/2}}, \ \ 9\log_2 p\right\}
\end{align*} 
The algorithm has reducer size $B(d,k)$, and with probability $1-2^{-d}$, it maps each input to at most 
\begin{equation*}
7\cdot \max\left\{d \ln 2,\ \ 6 \cdot 2^{\ceil{r/2}}\ln\left(\frac1{1-\delta}\right) \cdot  \left(\frac{d}{k}\right)^{\ceil{r/2}}\right\}
\end{equation*}
processors.
Moreover, with $\delta=1-1/n^3$, the edge-covering of the
input is error-free with probability at least $1-1/n$.  
\item  
There is constant $c_\delta>0$ such that for every $r,n,d,p,\delta$ with 
$r\le \sqrt{d/2}$, $n \geq d^{r\log_2 d}$,
and $\delta\ge 4/\sqrt{d}$,
\begin{equation*}\ohbound\ge (c_\delta/r)^r\cdot (d/\log^2_d n)^{\gamma r/2}.\end{equation*}
Further, if $\gamma r\le 1$ then
\begin{equation*}\ohbound\ge (c_\delta/r)^r\cdot (d/\log_d n)^{\gamma r}.\end{equation*}
\item  The bounds given in (a) and (b) also hold with $\ohbound$ replaced by $\eqohbound$.
\end{enumerate}
\end{theorem}

Observe that our lower bound in (b) (in its range of applicability) is
qualitatively very similar to the upper bound in (a) with maximal choice of
$k$, especially as the number of processors approaches the set size.
The only previous lower bound on replication rate or overhead for similarity
join \cite{afrati-dist} used the edge-isoperimetric bound for
$r=1$ to derive a lower bound of $(\log_2 n)/\log_2 q$ 
where $q$ is an upper bound on the reducer size for input sets of size
$n$ that are dense subsets of a subcubes;
however, for $q$ as large as $n/p$, the regime in which our bounds apply, 
this only yields a nearly trivial lower bound of $1/(1-\gamma)$.

We outline the upper bound improvements compared to previous
algorithms~\cite{afrati-fuzzy, afrati-dist, afrati-anchor}. 
Our upper bound with maximal $k$ comes from an algorithm with overhead $O(2d/\log_m (n/p))^{\ceil{r/2}}$. 

This improves on the
Ball-hashing-2 algorithm~\cite{afrati-fuzzy}, which has overhead
$(2d/r)^{\ceil{r/2}}$ using reducers of size $B(d,\ceil{r/2})$. Our algorithm also improves over the 
{\sc AnchorPoints} algorithm in~\cite{afrati-anchor}, which 
has overhead $O(B(d,\ceil{3r/2})/B(d,r))$, approximately $(d/r)^{\ceil{r/2}}$,
using reducers of size $B(d,\ceil{3r/2})$.  We finally note that the subcube algorithm~\cite{afrati-fuzzy,afrati-dist} has good overhead when $r=1$
and is not improved by our results, but it is not competitive for $r>1$.

Our algorithm is best for large $k$ with reducers of size $B(d,k) \approx n/p$. We improve the overhead
by up to a $(\frac{1}{2}\log_d (n/p))^{\ceil{r/2}}$ factor,
taking advantage of larger reducers to obtain better bounds.
For a comparison using a 
natural range of parameters, consider $d~=~\left(\log_2 \sqrt{n}\ \right)^2$ and $p = 2^{\sqrt{d}}=\sqrt{n}$. Our overhead for even $r$ scales like 
$(d\log_2 d)^{r/4}$, which improves over the best previous 
algorithm by a factor of roughly $(d/r^2 \log_2 d)^{r/4}$.  

Our algorithm is nearly-optimal for $n = 2^{\Theta(d)}$ and
$p = n^\gamma$ for constant $\gamma<1$.
In this case, choosing $\delta=1-1/2^{3d}$
we obtain an algorithm that is almost certainly correct and has overhead
only $2^{O(r)}d\log^{\ceil{r/2}} d$, since $\log_2 p\le d$.  This is a significant improvement over the best previous algorithm~\cite{afrati-anchor} which has overhead and replication rate 
$(d/r)^{\ceil{r/2}}$ in this regime.

%% file: 4-edge-iso-upper.tex
\section{Randomized Edge-Coverings using Edge-Isoperimetric Sets}\label{sec:edge-iso}

In this section, we describe our randomized edge-covering for $\G$ with overhead achieving the upper bound in Theorem~\ref{thm}.  This randomized edge-covering defines a one-round, data-parallel protocol for reporting all pairs in $S \subseteq \bits$ with Hamming distance at most $r$.  It thus also provides a MapReduce algorithm for Hamming similarity joins.   

We begin with our general protocol for arbitrary edge-transitive similarity
graphs, since it is easy to state and explains the intuition.  We then specialize this general protocol to Hamming distance and provide the specific analysis in this case.

\subsection*{General Edge-Transitive Graph Covering}
We build a randomized $(p,\delta)$-edge-covering for $G$ using unions of random translates of a fixed set $U^* \subseteq V$.
An \emph{edge-isoperimetric set} of size $s$ is any set $U^*(s)$ of vertices in $G$ that maximizes the number of edges among all induced subgraphs with $s$ vertices.  We will use translates of $U^* \deq U^*(s)$ for $s\le n/p$.
Let $\pi$ be a graph automorphism on $G$.  
Define $\pi(U^*)$ to be the set of
vertices $\pi(x)$ for each $x \in U^*$. 
By edge-transitivity, $\pi(U^*)$ contains the same number of edges as $U^*$ does.

To construct the $(p,\d)$-edge covering, we choose a parameter $L$ and
choose $Lp$ automorphisms $\pi_1,\ldots,\pi_{Lp}$ of $G$
uniformly at random, and define
\begin{equation*}A_i=\bigcup_{j\ \equiv\ i\mathrm{\ mod\ }p} \pi_j(U^*).
\end{equation*}
With $Lp$ large enough compared to the ratio of the number of edges in $G$ to
that of $U^*$, the sets $A_i$ cover a $\delta$-fraction of edges in $S$ with high probability. As we show, the fact that $|U^*|\le n/p$ ensures
a bounded overhead via Bernstein's concentration inequality. 

When interpreted as a one-round MapReduce algorithm, the construction above admits good bounds on both the maximum reducer size and the maximum number of times a vertex is replicated by the mappers. 

\begin{lemma}\label{lemma:reducer-size}  Let $U^*$ be a subset of vertices in $G=(V,E)$. The MapReduce algorithm that chooses $Lp$ uniformly random automorphisms $\pi_1,\ldots,\pi_{Lp}$ and maps $x \in V$ to the reducer assigned $\pi_j(U^*)$ whenever $x \in \pi_j(U^*)$  has reducer size $|U^*|$, and with probability $1 - |V|^{-1}$, it maps each vertex 
to at most $7 \cdot \max\{\ln |V|, \ Lp|U^*|/|V|\}$ reducers.  
\end{lemma}  
\begin{proof}
Any $x \in V$ is contained in $\pi_j(U^*)$ with probability $|U^*|/|V|$.  The expected number of indices $j$ such that $x \in \pi_j(U^*)$ is $Lp|U^*|/|V|$.  By independence of the automorphisms, a standard Chernoff bound implies that $x$ is mapped to more than $7 \cdot \max\{\ln |V|, \ Lp|U^*|/|V|\}$ reducers with probability at most $|V|^{-2}$.  Taking a union bound over $V$ finishes the proof.
\end{proof}

We now analyze the above algorithm for $\G$.  

\subsection*{Edge-Isoperimetry for Hamming Distance}

The classical edge-isoperimetric result for $\Gamma_{1}$ states that $s$ vertices
contain at most $\frac{1}{2}s\log_2 s$ edges~\cite{bernstein, harper, hart, lindsey}.  
The optimal set is a subcube when $\log_2 s$ is integral, and it is a natural interpolation between subcubes for other $s$ values.

For $\G$ with $r>1$, an exact edge-isoperimetric inequality is not known.
Surprisingly, optimal shapes must be far from subcubes.  For example, Hamming balls, which are essentially the optimal \emph{vertex-isoperimetric}
sets for $\Gamma_1$, are not only much better than subcubes, they are also
not too far from optimal edge-isoperimetric sets for $\G$.
We give an easy argument for very rough approximate optimality of Hamming
balls below and state a much sharper bound that was recently shown by the
second author and co-author~\cite{er:isoperimetry} in parallel work. 

We start with preliminaries about the number of edges in a Hamming ball.  We will lower bound the average degree in $\G$ of a ball of radius $k$, that is, the value $\er(\Ball( 0^d, k))$.

\begin{proposition}\label{ball-ub} For $k \in [d]$ with $\ceil{r/2} \leq k$, we have 
\[ 
\er(\Bk)\geq \frac{1}{2}\cdot \binom{k}{\ceil{r/2}} \cdot B(d-k,\floor{r/2})
\]
\end{proposition}

\begin{proof} 
Observe that $\er(\Bk)\ge \frac{1}{2} \displaystyle \min_{x\in \Bk} |\Bk\cap \Ball(x,r)|$.

We first argue that the minimum occurs when $|x|=k$.  Indeed, let $|x'|<k$
and $|x|=|x'|+1$ such that $x\oplus x'$ contains a single non-zero coordinate
$j$. \\\\
{\sc Claim:} $|\Bk\cap \Ball(x',r)|\ge |\Bk\cap\Ball(x,r))|.$\\\\
First note that for any $z\in \bits$, $\Ball(z,r)=z\oplus \Ball(0^d,r)$.
Let $y\in \Ball(0^d,r)$ have $y_j=0$ and let $y'$ be the same as $y$ except
that bit $y_j=1$.
Now $x'\oplus y=x\oplus y'$ and $x\oplus y=x'\oplus y'$.  Therefore, 
(i) the number of such $y$ for which $x'\oplus y\in \Bk$ but
$x\oplus y\notin \Bk$
is equal to (ii) the number of such $y'$ for which
$x\oplus y'\in \Bk$ but $x'\oplus y'\notin \Bk$.
Observe that (i) counts $\Bk\cap \Ball(x',r)\setminus (\Bk\cap \Ball(x,r))$
and (ii) counts a (strict) superset of
$\Bk\cap \Ball(x,r)\setminus (\Bk\cap \Ball(x',r))$ since for any
$z'\in \Ball(0^d,r)$ with $z'_j=1$, the element $z$ equal to $z'$ except that
$z_j=0$ is also a member of $\Ball(0^d,r)$.  (The strictness follows because
for those $y\in\Ball(0^d,r)$ with $|y|=k$ have $y'\notin\Ball(0^d,r)$.) 
This proves the claim which implies that the minimum occurs when $|x|=k$.

Now fix some $x\in \Bk$ with $|x|=k$.
We count vectors $z$ in $\Bk\cap \Ball(x,r)$. 
The vector $z$ may be obtained from $x$ by flipping $i$ bits from one to zero,
and $j$ bits from zero to one, as long as  
$i\le k$, $i + j \leq r$, and $j \leq i$.
Therefore, by summing over $(i,j)$ that meet these conditions, the number of 
vectors $z$ is at least
$
\displaystyle \sum_{i\le \min\{r,k\}}\sum_{j\le \min\{r-i,i\}} \binom{k}{i} \binom{d - k}{j}. \ 
$
This is a lower bound on the minimum, and thus average, degree of $\Bk$ in $\G$.
\end{proof}

\paragraph{Hamming Balls are nearly optimal for $r>1$.}  

We sketch the ideas behind the edge-isoperimetric inequality for $r > 1$.  Consider $A \subseteq \bits$.  Using standard down-shifting arguments, we may assume $A \subseteq \Ball(0^d, \floor{\log_2 |A|})$.  Thus, the average degree of $A$ in $\G$ is  roughly bounded by that of a ball of radius $\floor{\log_2 |A|}$.  In particular, $\er(A)\le B(d,\floor{r/2})\cdot B(\floor{\log_2 |A|},\ceil{r/2})$.  Ellis and Rashtchian~\cite{er:isoperimetry} improve this estimate by roughly a factor of $(\log_2 (d/\log_2 |A|))^{\floor{r/2}}$, nearly matching the Hamming ball upper bound. 

%
%
%
%

\subsection{Proof of Theorem~\ref{thm}(a)}

\ \ Now we prove the upper bound in our main theorem by exhibiting a randomized $(p,\d)$-edge-covering for sets of size $n$ in $\G$ achieving the claimed
overhead.
\newcommand{\pbound}{(d/\log_d n)^{r/2}}

For $p\le \pbound$ we can achieve the bound simply using the universal
algorithm.
Our construction for $p> \pbound$ consists of a partition of a
random set of Hamming balls. 
We denote the radius of these balls by $k$, which will be carefully chosen to  achieve efficient coverage and ensure proper load balancing.  
Intuitively, if $k$ is too small, then each ball has too few edges and we will
need too many balls to cover all the edges in $\G$.  
On the other hand, if $k$ is too large, then we risk having a large
intersection $|S \cap A_i|$ on input $S$. 
We will see that the value of $k$ that achieves the best overhead bounds is
such that $B(d,k) \approx |S|/p=n/p$, but we consider all values of $k$ for
which the algorithm makes sense.

\subsection*{Randomized Edge-Covering Construction using Hamming Balls} 

Let $n$ be the input set size and $k\ge \ceil{r/2}$ to be any integer such that $B(d,k)\le n/p$. 
Observe that by our assumption $p> \pbound$, we have that 
$k\le d/2-\Omega(\sqrt{rd\log d})$.   

The number of Hamming balls in the covering will be proportional to the ratio
of the number of edges in $\G$ and in $\Ball(0^d,k)$.  
Let this ratio $\lk$ be
\[ 
  \lk \deq \frac{|\Er(\bits)|}{|\Er(\Ball(0^d,k))|}.
\]
Define a random distribution $\Ak$ of $p$-tuples of subsets of $\bits$ with
$(A_1,\ldots,A_p) \sim \Ak$ chosen as follows:
Let $x_1,\ldots, x_{Lp}\subseteq \bits$ be a uniformly random sequence of
vectors where $L=\ln(1/(1-\delta))\cdot \ceil{\lk/p}$.  
Define the random sets $A_i$ as
\[
	A_i = \bigcup_{j\ \equiv\ i \mathrm{\ mod \ } p } \Ball(x_i,k).
\]

We begin by bounding $\lk$, the ratio of the number of edges in $\G$
to that in $\Ball(0^d,k)$.  
This bound on $\ell_k$  (and thus $Lp$) combines with
Lemma~\ref{lemma:reducer-size} to immediately provide the claimed upper
bounds in Theorem~\ref{thm} on the reducer size and the number of times
that a vector is replicated.  

\begin{proposition}\label{lk-ub} Fix $k \in [d]$ with $\ceil{r/2}\le k\le d/2-\ceil{r/2}.$
Then,
\[ 
   \frac{|\Er(\bits)|}{|\Er(\Ball(0^d,k))|} <
    \left(\frac{2d}{k}\right)^{\ceil{r/2}}\cdot \frac{2^{d+1}}{B(d,k)}.
\]
\end{proposition}

\begin{proof}
The definition of $\Er(\bits)$ and the lower bound on $\er(\Ball(0^d,k))$ in
Proposition~\ref{ball-ub} imply that
\begin{equation*}
   \lk = \frac{|\Er(\bits)|}{|\Er(\Ball(0^d,k))|}\le
\frac{B(d,r)2^{d}}{ \binom{k}{\ceil{r/2}} \cdot B(d-k,\floor{r/2}) \cdot B(d,k)}.
\end{equation*}
We invoke Proposition~\ref{ballbound} from the appendix and use the assumption $d/2\ge 2\ceil{r/2}$ to bound $B(d,r)/B(d,\floor{r/2})\le \frac{5}{4}d^{\ceil{r/2}}/r^{\floor{r/2}}$.  Then, combining this with the simple bound $\binom{k}{j}\ge (k/j)^j$ we obtain that the upper bound on $\lk$ is at most 
\begin{equation*} 
	\frac{5}{4}
	\cdot
	\frac{(\ceil{r/2})^{\ceil{r/2}}}{r^{\floor{r/2}}}
	\cdot 
	\frac{B(d,\floor{r/2})}{B(d-k,\floor{r/2})}
	\cdot
	\left(\frac{d}{k}\right)^{\ceil{r/2}}
	\cdot
	\frac{2^d}{B(d,k)}
	\ \ <  \ \ 
	\frac{d^{\floor{r/2}}}{(d-k)^{\floor{r/2}}}
	\cdot
	\left(\frac{d}{k}\right)^{\ceil{r/2}}\cdot\frac{2^{d+1}}{B(d,k)}.
\end{equation*}
We used that $(\ceil{r/2})^{\ceil{r/2}}/r^{\floor{r/2}}$ is at most 1 with $r$ even
and achieves its maximum value of 1.35 for $r$ odd at $r=5$. 
Finally, since $d-k-\floor{r/2}\ge d/2$, this is at most
$(2d/k)^{\ceil{r/2}}\cdot 2^{d+1}/B(d,k)$.
\end{proof}

The upper bound on the overhead in Theorem~\ref{thm} follows from the following two lemmas, and the fact that $B(d,k)\le d^k$, which implies that $B(d,k)$ for 
$k~\le~(1-\gamma)\log_d n~=~\log_d (n/p)$ is at most $n/p$.

\begin{lemma} For all $S \subseteq \bits$ with $|S|=n$, for $k\ge \ceil{r/2}$ and
$B(d,k)\le n/p$, we have
\begin{equation*}
\Ebb[\o(\Ak,S)] \leq \max\{6 \ln(1/(1-\delta)) \cdot  (2d/k)^{\ceil{r/2}}, 9\log_2 p\}.
\end{equation*}
Moreover, the bound on the overhead holds almost surely.
\end{lemma}

\begin{proof} 
Let $S\subseteq \bits$ with $|S|=n$.
We will upper bound $\Ebb[\max_i |S \cap A_i|]$ by proving an upper bound on
$\Ebb[|S \cap A_i|]$ for each $i\in [p]$, where the latter expectation is over the distribution on $A_i$ induced by $\Ak$.  After doing so, we union bound over $[p]$ to guarantee the bound on the expected max with high probability. 

Fix $i \in [p]$.  The set $A_i$ is the union of
$\Ball(x_j,k)$ for $L$ vectors $x_j$ chosen uniformly and independently from
$\bits$.
Define the random variable $Z_j~=~|S~\cap~\Ball(x_j,k)|$ for $j \in [L]$.
Then, defining $Z=\sum_j Z_j$, we have
$
	|S\cap A_i| \le \sum_{j} Z_j = Z.
$
We first upper bound $\Ebb[Z] \ge \Ebb[|S \cap A_i|]$.  
Observe that for a uniformly random $x_j$ from $\bits$, every element
of $\bits$ is equally likely to be in $\Ball(x_j,k)$.   Therefore every $j \in [L]$,
\begin{equation*}
	\Ebb[Z_j] = \Ebb [|S \cap \Ball(x_j,k)|]
	= \sum_{y\in S} \Pr[y\in \Ball(x_j,k)] =\sum_{y\in S} \frac{B(d,k)}{2^d}
	= \frac{n\cdot B(d,k)}{2^d}.
\end{equation*}
Thus, recalling $L=\ln(1/(1-\delta))\cdot \ceil{\lk/p}$, and using the upper bound on $\lk$ (the ratio of the number of edges in  
$\Ball(0^d,k)$ and in $\G$)  from Proposition~\ref{lk-ub}, we compute
\begin{equation*}
	\Ebb[Z]
	\ =\  L\cdot \frac{n\cdot B(d,k)}{2^d}
	\ =\  \ln(1/(1-\delta))\cdot\ceil{\lk/p} \cdot \frac{n\cdot B(d,k)}{2^d}
	\ <\ 3\ln(1/(1-\delta))\cdot (2d/k)^{\ceil{r/2}}\cdot \frac{n}{p}.
\end{equation*}
For concentration, we use Bernstein's inequality (e.g., Thm. 3.6 in~\cite{chunglu}), which applies since $0 \leq Z_j \leq B(d,k)$ and $\Ebb[Z_j^2]\le B(d,k)\Ebb[Z_j]$.  Therefore,  for any $\lambda\ge \Ebb[Z]$,
\begin{equation*}
	\Pr\left[Z > 2 \lambda \right]
	\le e^{-\frac{1}{2}\lambda^2/(B(d,k)\Ebb[Z]+\lambda B(d,k)/3)}
	\le e^{-\frac{3}{8}\lambda/B(d,k)}.
\end{equation*}
For every $\lambda=\lambda(\theta)\ge \max\{\Ebb[Z],4 B(d,k)(\log_2 p+\theta)\}$
the probability that $|S\cap A_i|\ge 2\lambda$ is at most $e^{-\theta}/p^{3/2}$. 
By the union bound over $[p]$, the probability that $\max_i~|S~\cap~A_i|$ exceeds $2\lambda$ is at most
$e^{-\theta}/p^{1/2}$.   In particular, we have an upper bound on $\o(\Ak,S)$
of $2\lambda$. 
\end{proof}
We conclude with the correctness argument.
\begin{lemma}
For any $n$, $\Ak$ forms a $(p,\delta)$-edge-covering of the
$n$-subsets of $\bits$. Further, if $\delta$ is chosen to be $1-1/n^3$,
then $(A_1,\ldots,A_p)\sim\Ak$ covers all edges of any fixed input set of size
$n$ with probability
at least $1-1/n$.
\end{lemma}

\begin{proof}
Fix $S\in \bits$ with $|S|=n$.
Since $\G$ is edge-transitive, for $x_j\in \bits$ chosen uniformly at 
random, each pair $\{u,v\}\in S$ of distance at most $r$ satisfies
$\set{u,v} \subseteq \Ball(x_j,k)$ with probability 
\[
	1/\lk=|\Er(\Ball(0^d,k))|/|\Er(\bits)|.
\]
Therefore, the probability that a fixed $\set{u,v}$ is not covered by 
$(A_1,\ldots,A_p)\in\Ak$ is at most
\begin{equation*}
(1-1/\lk)^{Lp}\le (1-1/\lk)^{\ln(1/(1-\delta))\lk} \le e^{-\ln(1/(1-\delta))}=1-\delta.
\end{equation*}
Each edge of $\Er(S)$ is covered with probability at least $\delta$,
and in expectation $(A_1,\ldots,A_p)\sim\Ak$ covers at least a $\delta$
fraction of $\Er(S)$.

There are at most $n^2$ pairs in $S$ so with $\delta=1-1/n^3$ (in which
case $\Ak$ chooses $3p\ceil{\lk/p}\ln n$ uniformly random balls of radius $k$)
a union bound implies that
the probability that $(A_1,\ldots,A_p)\sim\Ak$ covers all pairs in $S$ is at
least $1-1/n$.  
\end{proof}

%% file: 5-lower-bound.tex
\section{The Lower Bound}\label{sec:lb}

We now prove the lower bound of Theorem~\ref{mainlb}(b).
Let $\A$ be any randomized $(p,\delta)$-edge cover of the $n$-subsets of
$\bits$ for $\G$.

We will use a variant of Yao's minimax principle by exhibiting a ``hard''
distribution on subsets $S$ of $\bits$ of size $n$.  The usual form of the
argument would then be to show that the distribution has the
property that the expected overhead of any  fixed 
$(A_1,\ldots,A_p)$ in the support of $\A$ 
is large and conclude that the expected overhead for $\A$ on
some element of the support of the hard distribution is large.

However, this does not quite work.
Since $\A$ produces a large fraction of the
similar pairs of $S$
\emph{only in expectation}, its
support may contain $(A_1,\ldots,A_p)$ that achieve low overhead but without
producing many similar pairs. 
Instead, we show that (1) the expected overhead 
under the hard distribution must be large
for any tuple $(A_1,\ldots,A_p)$ such that $|\bigcup_{i\in [p]} \Er(A_i)|$ is 
at least a $\delta/2$ fraction of $|\edges(\G)|$
and that (2) any $(A_1,\ldots,A_p)$ that fails to have this property does badly at covering subsets $S$ chosen according
to this distribution and hence such tuples must only be a small fraction of the
probability mass of $\A$.

The hard input distribution $\Dn$ over sets $S~\subseteq~\bits$ of size $n$ will
choose $S$ to be a random sub-sampled Hamming ball of suitable size so that $S$ has relatively high density within that ball. 

\begin{definition}
\label{defn-Dn}\label{Dn} 
Let $d \in \Z^+$ and $r \in [d]$ be positive integers.    
Given $n \in [2^{d-1}]$, let $R = R(n,d,r) \in \Z^+$ denote the unique positive integer such that $r$~divides~$R$ 
and  \[B(d,R-r) < n \leq B(d,R).\] 
Define distribution $\Dn$ on subsets $S$ of $\bits$ of size $n$ by first
choosing $x \in \bits$ uniformly at random, then choosing a random 
subset of $n$ elements of $\Ball(x,R)$.
\end{definition}

Using the fact that $\Dn$ samples edges uniformly from $\G$ we have the 
following simple property of potential covers.

\begin{proposition}
\label{mostlybig}
If $(A_1,\ldots,A_p)$ satisfies 
$\displaystyle |\bigcup_{i\in [p]} \Er(A_i)|\le \alpha\cdot |\edges(\G)|
$
for some $\a \in [0,1]$, then 
\begin{equation*}
\E_{S\sim \Dn}\left|\bigcup_{i\in [p]} \Er(A_i\cap S)\right| \ \le \ \alpha\cdot |\Er(S)|.
\end{equation*}
\end{proposition}

\begin{proof}
Observe that
$
	\bigcup_{i\in [p]} \Er(A_i\cap S)=\left(\bigcup_{i\in [p]}\Er(A_i)\right)\cap \Er(S).
$
For $S$ chosen according to $\Dn$, the set $\Er(S)$ contains each edge of
$\G$ with equal probability, so the claim follows by linearity of expectation.  
\end{proof}

Observing that $\ohbound$ is precisely 
\[
	\ohbound = \min_{\A}\ \max_{S\subseteq\bits:\ |S|=n} \ \o(\A,S),
\]
Yao's minimax principle~\cite{yao} applied to distribution $\Dn$ yields the following.

\begin{lemma}
\label{yao}
Let $\delta\in [0,1]$ and $A^{good}_\delta$ be the set of tuples
$(A_1,\ldots, A_p)$ of subsets of $\bits$
such that 
$|\bigcup_{i\in [p]} \Er(A_i)|\ge \frac{\delta}{2}\cdot |\edges(\G)|$.  Then,
\begin{equation*}
\ohbound \ge\ \frac{\delta}{2}\cdot \max_{(A_1,\ldots,A_p)\in A^{good}_\delta} \ \ES \left[ \max_{i} |A_i \cap S|\cdot \frac{p}{n}\right].
\end{equation*}
\end{lemma}

\begin{proof}
Let $\A$ be a randomized $(p,\delta)$-edge cover for the $n$ subsets of
$\G$.
Therefore 
the expectation over $(A_1,\ldots,A_p)\sim \A$ and $S\sim \Dn$
of $|\bigcup_{i\in [p]} \Er(A_i\cap S)|$ must be at least
$\delta |\Er(S)|$, since
the $\delta |\Er(S)|$ bound holds for every choice of $S$.
By Proposition~\ref{mostlybig} for $(A_1,\ldots,A_p)\notin A^{good}_\delta$
we have
$\E_{S\sim \Dn} |\bigcup_{i\in [p]} \Er(A_i\cap S)|\le \delta |\Er(S)|/2$.
Since the most this quantity can be is $|\Er(S)|$, by Markov's inequality, the
probability that $\A$ produces
$(A_1,\ldots,A_p)$ that is in $A^{good}_\delta$ is at least $\delta/2$ and
since all quantities are non-negative, the lower bound follows.
\end{proof}

For the remainder of this section we fix $(A_1,\ldots,A_p)$ in
$A^{good}_\delta$, so that 
\begin{equation}
\label{min-edge-density}
|\bigcup_{i\in [p]} \Er(A_i)|\ge \frac{\delta}{2}|\edges(\G)|=\delta B(d,r)2^{d-2}.\ \
\end{equation}
We prove that $\ES \ \max_{i} |A_i \cap S|$ must be large. Note that for sufficiently small $r$, and sufficiently large
$\delta$, equation (\ref{min-edge-density}) implies a lower bound
on the total number of edges of $\Gdr$ covered by the $A_i$.

\begin{proposition}
\label{exact-edge-density}
Suppose that $r\le \sqrt{d/2}$ and $\delta\ge 4/\sqrt{d}$.
If $(A_1,\ldots,A_p)\in A^{good}_\delta$, then 
\begin{equation*}
|\bigcup_{i\in [p]} \Edr(A_i)|\ge \frac{\delta}{4}\cdot \binom{d}{r}\cdot 2^{d-1}
\end{equation*}
\end{proposition}

\begin{proof}
There are $B(d,r-1) 2^{d-1}$ edges of $\G$ that are not in $\Gdr$.
For $r\le \sqrt{d/2}$, we see that
\[ \frac{\binom{d}{r-1}}{\binom{d}{r}}=\frac{r}{d-r+1}<\frac{1}{\sqrt{d}+1}\] and hence
$B(d,r-1)\le B(d,r)/\sqrt{d}\le \frac{\delta}{4}\cdot B(d,r) $.
Combining with (\ref{min-edge-density}), we have that at
least
$\frac{\delta}{4} B(d,r) 2^{d-1}\ge \frac{\delta}{4}\binom{d}{r}2^{d-1}$ edges
of $\Gdr$ are contained in 
$\bigcup_{i\in [p]} \Edr(A_i)$.
\end{proof}

Next, we reduce the problem to reasoning about covers that only include necessary elements. Intuitively, it suffices to cover each vertex with fewer sets than its degree. 

\begin{definition}
A tuple $(A_1,\ldots,A_p)$ is $r$-\emph{pruned} if for every
$x\in \bits$, $|\set{i:\ x\in A_i}|\le B(d,r)-1$.
\end{definition}

\begin{lemma}
\label{pruninglemma}
For any $(A_1,\ldots,A_p)$ there exists a $r$-pruned
$(A'_1,\ldots, A'_p)$ such that $A'_i\subseteq A_i$ for all $i\in [p]$ and
$\bigcup_{i\in [p]} \Er(A'_i)=\bigcup_{i\in [p]} \Er(A_i)$.
\end{lemma}

\begin{proof}
Define $\w(x) = |\{ i \in [p]\ :\  x \in A_i\}|$.
For $i\in [p]$, say that a set $A_i$ is {\em pivotal with respect to $x$} if
there exists an edge $\{x,y\}$ in $\G$ with $x \neq y$ such 
that $\{x,y\} \subseteq A_i$ while $\{x,y\} \subsetneq A_j$  for all $j < i$. 
For each $x\in \bits$, remove $x$ from all but the pivotal sets for $x$. 
Let $(A'_1,\ldots,A'_p)$ be the resulting tuple of sets.
Clearly, $\bigcup_i \Er(A'_i)=\bigcup_i \Er(A_i)$ by construction.
For each $x$, the number of pivotal sets for is at most the size of the open
neighborhood of $x$ in $\G$, which is $B(d,r) - 1$.   Therefore, 
$w(x)\le B(d,r)-1$ for all $x$, and hence, $(A'_1,\ldots,A'_p)$ is
$r$-pruned.
\end{proof}

With these basic preliminaries out of the way we describe the overall proof
strategy.

\subsubsection*{Proof Strategy} 

Lemma~\ref{pruninglemma} allows us to assume without loss of generality that
$(A_1,\ldots,A_p)$ is $r$-pruned since pruning yields the same set of
similar pairs covered and does not increase overhead.

Our argument capitalizes on the following tradeoff. We have two cases:
For the first case, if $\sum_{i\in [p]} |A_i|$ is much larger than $2^d$ then
the algorithm will replicate vertices many times, implying a large overhead. 
For the second case, if $\sum_{i\in [p]} |A_i|$ is closer to $2^d$,
then we will prove that 
the average edge density of induced edges of $\G$ in the sets $A_i$ must be
large. In other words, in the second case, the weighted average of $\er(A_i)$
is large.  By Proposition~\ref{exact-edge-density}, this means that the
weighted average of $\edr(A_i)$ is also large.

Using a key combinatorial lemma which shows that sets $A$ with
large $\edr(A)$ must also have large $\eR(A)$ where $R$ is
the distance defined by distribution $\Dn$ and hence, in the second case,
the sum of the $|\ER(A_i)|$ is large.

Finally, it is not hard to see that $|\ER(A_i)|$ being large is equivalent to
saying that the expected size of $|A_i\cap \Ball(x,R)|$ is large for $x$ a
randomly chosen element of $A_i$ and it is not hard to see that a similar
property also applies to high density random subsets $S$ of $\Ball(x,R)$;
like elements chosen according to $\Dn$.

This roughly suggests that a set $S$ chosen from $\Dn$ will end up producing a
large overhead at some set $A_i$ that contains 
the center $x$ of the ball $\Ball(x,R)$ used to define $S$.  
However, as $x$ varies, which of these sets $A_i$ are possible also varies
and the choice of $x$ and the set index $i$ are correlated so choosing just
one $A_i$ would not let the rough argument succeed.
Instead, using the fact that, without loss of generality, no element of $\bits$
needs to appear in too many sets $A_i$, we show that we can eliminate the bias
by lower bounding the maximum overhead by the average overhead over all sets
$A_i$ with $x\in A_i$.

We begin by proving the last part of this strategy; i.e., we
first relate the expected maximum size of $A_i\cap S$ to the number of pairs
with distance at most $R$ in $A_i$.
\begin{lemma}\label{lemma:max-to-ER} 
Let $d\in \Z^+$, $r\in [d]$, $n\in [2^{d-1}]$ and $R=R(n,d,r)=kr\le d/2$
for integer $k\ge 2$.
Then for any $r$-pruned $\set{A_1,\ldots,A_p}$ with each $A_i\subseteq \bits$,
\begin{equation*}
\ES\left[\ \max_{i\in [p]} |A_i \cap S| \right] \   >\ 
\frac{B(d,R-r)}{2^{d-1}B(d,R)(B(d,r)-1)} \sum_{i\in [p]} |\ER(A_i)|
\ge \frac{R^{(r)}r!}{d^{2r}}\sum_{i\in [p]} |\ER(A_i)|/2^d.
\end{equation*}
\end{lemma}

\begin{proof} 
The distribution $\Dn$ first chooses a center $x~\in~\bits$ uniformly at random
and then chooses $S~\subseteq~\Ball(x,R)$ uniformly at random from all subsets of
size $n$.  Let $\Dn^x$ denote the conditional distribution of $S$ given a
fixed choice of $x$.  
Then, since $\set{A_1,\ldots,A_p}$, for all
$x\in \bits$, we have \[\w(x) = |\{ i \in [p]\ :\  x \in A_i\}|\le B(d,r)-1.\]
We can restrict our choice of maximum to sets containing the center $x$
and lower bound the maximum over such sets by their average to obtain
\begin{align*}
\ES\left[\max_{i\in [p]} |A_i \cap S| \right ]  
&= \frac{1}{2^d} \sum_{x\in \bits} \E_{S\sim \Dn^x}\left[\max_{i\in[p]} |A_i \cap S| \right]  \\
&\ge \frac{1}{2^d} \sum_{x\in \bigcup_{i} A_i} \E_{S\sim \Dn^x}\left[\max_{i\in[p]} |A_i \cap S| \right]  \\
&\ge 
\frac{1}{2^d} \sum_{x\in \bigcup_{i} A_i} \E_{S\sim \Dn^x}\left[\max_{i:\ x\in A_i} |A_i \cap S| \right]  \\
&\ge \frac{1}{2^d} \sum_{x\in \bigcup_{i} A_i} \E_{S\sim \Dn^x}\left[\sum_{i:\ x\in A_i} \frac{|A_i \cap S|}{B(d,r)-1} \right]\qquad\mbox{\ \ since }0<w(x)\le B(d,r)-1\\
&= \frac{1}{2^d} \sum_{x\in \bigcup_{i} A_i} \E_{S\sim \Dn^x}\left[\sum_{i\in [p]}\mathbb{1}_{x\in A_i} \cdot \frac{|A_i \cap S|}{B(d,r)-1} \right]\\
&= \frac{1}{2^d} \sum_{x\in \bigcup_{i} A_i} \sum_{i\in [p]} \mathbb{1}_{x\in A_i} \cdot
\E_{S\sim \Dn^x}\left[\frac{|A_i\cap S|}{B(d,r)-1} \right]\\
&= \frac{1}{2^d(B(d,r)-1)} \sum_{i\in [p]} \sum_{x\in A_i} 
\E_{S\sim \Dn^x}\left[|A_i\cap S| \right].\numberthis\label{fact1}
\end{align*}
Now for each $i\in [p]$ and $x\in A_i$, we have $S\subseteq \Ball(x,R)$; therefore
\begin{align*}
\sum_{x\in A_i} \E_{S\sim \Dn^x}\left[|A_i\cap S|\right]
&= \sum_{x\in A_i} \E_{S\sim \Dn^x}\left[\sum_{y\in \Ball(x,R)\cap A_i} \mathbb{1}_{y\in S}\right]\\
&= \sum_{x\in A_i} \sum_{y\in \Ball(x,R)\cap A_i} \E_{S\sim\Dn^x}\left[\mathbb{1}_{y\in S}\right]\\
&= \sum_{x\in A_i} \sum_{y\in \Ball(x,R)\cap A_i} \frac{n}{B(d,R)}\\
&= \frac{n}{B(d,R)} \cdot 2|\ER(A_i)|
\end{align*}
since the double summation counts each pair $\set{x,y}$ in $A_i$ of Hamming
distance at most $R$ exactly twice.
Now, by construction of distribution $\Dn$, $n$ is larger than $B(d,R-r)$,
so inserting this bound in (\ref{fact1}) yields that
$\ES\left[\max_{i} |A_i \cap S| \right ]$ is larger than
$\frac{B(d,R-r)}{2^{d-1}B(d,R)(B(d,r)-1)} \sum_{i\in [p]} |\ER(A_i)|$
as required.
Now, $B(d,r)\le d^r/r!$ and, since $R>r$, Proposition~\ref{ballbound}
proved in the appendix shows that
$B(d,R)/B(d,R-r)\le (d+1)d^{r-1}/R^{(r)}$ which is at most $2d^r/R^{(r)}$.
Together these yield the final claimed bound.
\end{proof}

To apply Lemma~\ref{lemma:max-to-ER}, we need a lower bound on
$|\ER(A_i)|$.  We obtain such a lower bound by proving a purely
graph-theoretic result.  This result relates  
different distances in subsets of the hypercube.
More precisely, we prove that subsets of $\bits$ with sufficiently
large density of Hamming distance $r$ pairs (i.e., sufficiently dense induced
subgraphs of $\Gdr$) also have a relatively large
density of pairs at Hamming distance at most $R$ when $R$ is a multiple of $r$.
This means that we can reduce the problem of lower bounding $|\ER(A_i)|$
to that of lower bounding $\edr(A_i)$.  
We give this lemma and its proof in the next subsection.

\subsection{Relating distance $r$ density and distance $R$ density on the hypercube}

\begin{lemma}\label{cor:path-lb}
\fixb.
Assume that $A \subseteq \bits$ satisfies $\edr(A) 
\geq 4\binom{R-r}{b+1}\binom{d}{r-b-1}$.
Then
\[ 
	\eR(A) \geq \frac{\edr(A)^{R/r}}{\denom}.
\]
\end{lemma}

This lemma is nearly tight; examples include Hamming balls of varying
radii and subcubes of varying dimensionality.  Some lower bound on $\edr(A)$
is necessary for this to hold since unions of well-separated balls of radius
$r$ have many pairs at distance up to $r$, but there are no pairs precisely
at distances between $2r$ and $R$.

The overall approach to proving Lemma~\ref{cor:path-lb} is to define a certain
class of paths of total length
$R/r$ in the induced subgraph $\Gdr[A]$ and prove that this set of paths is
large.  The start and end of each such path yields a pair vertices of $A$
at distance at most $R$ in the hypercube.
We then show that any pair of vertices in the hypercube can be connected by
relatively few such paths, yielding a large lower bound on the total number
of pairs of vertices in $A$ of distance at most $R$.

\begin{defn}  \fixb. An \emph{$(R,b)$-path} is a sequence $(v_0,v_1,\ldots,v_{R/r}) \in \{0,1\}^{d \times (R/r+1)}$ with  $\dist(v_{j-1},v_j) = r$ and 
$\dist(v_0,v_j) \geq j(r-2b)$ for every $j \in [R/r]$.  Let $\pi_{R,b}(A)$ denote the 
number of $(R,b)$-paths with all vectors in $A\subseteq \bits$. 
\end{defn}

Note that an $(R,b)$-path generalizes a $\Gdr$ shortest path, which is an
$(R,0)$-path. 
Lemma~\ref{cor:path-lb} is the immediate consequence of the following two
lemmas.

\begin{lemma}\label{lem:sid} \fixb.  Let $A \subseteq \bits$ be a subset and define $N=|A|$ and $M=|\Edr(A)|$.  If 
$ \frac{M}{N} \geq 4\binom{R-r}{b+1}\binom{d}{r - b -1}$,
then 
\[
	\pi_{R,b}(A) \geq N\left(\frac{M}{4N}\right)^{R/r}.
\]
\end{lemma}

\begin{lemma}\label{claim:rb-vs-ER} \fixb.  For any $A \subseteq \bits$,
\[
	|\ER(A)| \geq \frac{\pi_{R,b}(A)}{R! \cdot d^{bR/r}}.
\]
\end{lemma}
More precisely, from Lemmas~\ref{lem:sid} and~\ref{claim:rb-vs-ER} we have the following,  which yields Lemma~\ref{cor:path-lb},
\begin{equation*}
	\eR(A) \ = \  \frac{|\ER(A)|}{|A|} \  \geq \ \frac{\pi_{R,b}(A)}{R! \cdot d^{Rb/r}\cdot |A|} \ \geq \ \frac{e_r(A)^{R/r}}
	{4^{R/r}\cdot R!\cdot d^{Rb/r}}.
\end{equation*}

We begin by proving Lemma~\ref{lem:sid};
its proof is similar to proofs of lower bounds on the number of walks in a
graph given by Alon and Rusza~\cite{alon-rusza} and
Katz and Tao~\cite{katz-tao}.

\begin{proof}[Proof of Lemma~\ref{lem:sid}]
We prove the bound by induction on $N$.  For $N \leq 2$, the bound holds trivially.  Let $\d^*$ denote the 
minimum degree in the subgraph of $\Gdr$ induced by $A$.  We first consider the case with $\d^* \geq M/(2N)$.  
We will count $(R,b)$-paths $(v_0,\ldots,v_{R/r})$ by lower bounding the number of choices for $v_{j+1}$ given 
the path up to $v_j$.   We will argue that

\begin{equation}\label{rb-path-product}
\pi_{R,b}(A)\ \geq \ \d^*N\prod_{j=1}^{R/r - 1}\left[\d^* - \binom{jr}{b+1} \binom{d}{r - b - 1}\right].
\end{equation}
There are $\d^*N$ possibilities for $v_0$.  For $j = 1, 2, \ldots, R/r - 1$,
we say that a vertex $v_{j+1}$ is {\em bad} if 
\[
	\dist(v_0, v_{j+1}) < \dist(v_0, v_{j}) + (r-2b). 
\]
We must exclude the bad neighbors of $v_j$ when choosing $v_{j+1}$.  Each neighbor $v_{j+1}$ of $v_j$ in $\Gdr$ differs from $v_j$ in exactly $r$ coordinates.
Crucially,  for $v_{j+1}$ to be bad, at least $b+1$ of the differing
coordinates must be in the support of $v_0 \oplus v_j$ because each 
such coordinate used in this step reduces the distance of $v_{j+1}$
from $v_0$ by exactly 2.
Since $|v_0 \oplus v_j| = \dist(v_0,v_j) \leq jr$, and the other $r - b -1$
coordinates are arbitrary, there are at most 
$
	\binom{jr}{b+1} \binom{d}{r - b - 1}
$
bad neighbors of $v_j$ in $\Gdr$.  We have assumed that $\d^*$ satisfies
\begin{equation*}
	\d^* \geq \frac{M}{2N}  
	\geq 
	2\binom{R-r}{b+1} \binom{d}{r - b - 1}
	\geq 
	2\binom{jr}{b+1} \binom{d}{r - b - 1}.
\end{equation*}
Since each term in the product in (\ref{rb-path-product}) is at least $\d^*/2$, 
we conclude $\pi_{R,b}(A)\geq N\left(\frac{M}{4N}\right)^{R/r}$.

Now assume that the minimum degree $\d^*$ is less than $M/(2N)$. Remove a vertex with degree $\d^*$.  The 
resulting graph on $N-1$ vertices has larger average degree.   By induction, the number of $(R,b)$-paths on $N-1$ 
vertices is at least
\[ 
	(N-1)\left(\frac{M-\d^*}{4(N-1)} \right)^{R/r}.  
\]
We claim this is at least $N(M/(4N))^{R/r}$.  Rearranging both sides,  we need to show
\[
	\left(\frac{M-\d^*}{M} \right)^{(R/r)/((R/r)-1)} 
	\geq \frac{N-1}{N}.
\]
By the assumptions $\d^* < M/(2N)$ and $R/r \geq 2$, we have, as desired, 
\begin{equation*}
	\left(1-\frac{\d^*}{M} \right)^{(R/r)/((R/r)-1)}
	>  \left(1 - \frac{1}{2N}  \right)^{(R/r)/((R/r)-1)}
	\geq \left( 1 - \frac{1}{2N} \right)^2 
	> 1- \frac{1}{N}.
\end{equation*}
\end{proof}

We now finish the proof of Lemma~\ref{cor:path-lb} by proving
Lemma~\ref{claim:rb-vs-ER} which says that the number of pairs of distance
$\le R$ in $A$ is relatively large compared to the
number of $(R,b)$-paths in $A$.

\begin{proof}[Proof of Lemma~\ref{claim:rb-vs-ER}]
Consider a pair $(u,v) \in A \times A$ with $\dist(u,v) \leq R$.     
We claim that the number of $(R,b)$-paths  between $u$ and $v$ is at most $R! \cdot d^{bR/r}$. 
Such a path is specified by a sequence of edges in $\Gdr$, and once we fix $u$ and $v$,  this sequence is specified 
by 
\begin{enumerate}\item a size $R$ multiset $M$ containing coordinates in~$[d]$ differing between 
adjacent vertices, \item a partition $\Pi$ of $M$ into $R/r$ sets of size $r$. \end{enumerate}  
We claim that there are at most $d^{bR/r}$ choices for $M$ and $R!$ choices for $\Pi$.  We simply upper bound the 
count for $\Pi$ by the number of permutations of $R$ elements. Moving on to $M$, let $k = \dist(u,v)$ and recall 
that the $(R,b)$-path definition requires $k \geq R - 2bR/r$.  Observe that $M$ contains the $k$ elements in the 
support of $u \oplus v$.  Additionally, $M$ contains $(R -k)$ elements that come in identical pairs.  That is, they form 
a partition into $(R -k)/2$ pairs of identical elements.  Once $u \oplus v$ is fixed, these pairs determine $M$, and 
there are at most $d^{(R -k)/2}\leq d^{bR/r}$ choices for the pairs. 
\end{proof}

\subsection{Proof of Theorem~\ref{mainlb}(b)}

As described in our overall proof strategy, our lower bound involves
a tradeoff between the average number of times 
$t = \sum_{i\in [p]} |A_i|/2^d$ that elements of $\bits$ are
covered by $\set{A_1,\ldots,A_p}$ and the edge density of the individual $A_i$,
and hence the size of the intersections of sets $S\sim \Dn$, which are
similar to Hamming balls in structure.
The function $g$ that we define below helps us capture this overhead tradeoff.

\begin{definition}
Define $g:[1,\infty) \to \R^+$ by \ \ 
$ \displaystyle
	g(\tau) =  \min_{\substack{(A_1,\ldots,A_p)\in A^{good}_\delta\\ \sum_i |A_i| \ \leq \ \tau\cdot 2^d}}\  \ES\left[ \max_i |A_i \cap S|\right]\cdot \frac{p}{n}
$
\end{definition}

In particular, our overhead tradeoff is encapsulated in the
following lemma.

\begin{lemma}\label{prop:g(t)-lower-bound} 
Let $(A_1,\ldots,A_p)\in A^{good}_\delta$ and let
$t = 2^{-d} \sum_{i} |A_i|$.  Then,
\[ 
	\ES [\ \max_{i} |A_i \cap S| \ ] \ \geq\ \max \{t,g(t)\}\cdot \frac{n}{p}
\]
\end{lemma}

\begin{proof}  
We first prove the result for $\ohbound$.
The bound in terms of $g(t)$ follows immediately from the definition of $g$. 
For the bound in terms of $t$, observe that, since each $x \in \bits$ 
has $\mathbf{Pr}_S[x \in S] = n2^{-d}$, by linearity of expectation we have
\[
	\ES [\ \max_{i} |A_i \cap S| \ ]
	\ \geq \ 
	\ES \left[\frac{1}{p} \sum_{i\in[p]} |A_i \cap S| \right] 
	\ = \ 
	\frac{1}{p} \sum_{i\in[p]} |A_i|\cdot\frac{n}{2^d}
	\ = \ 
	t\cdot \frac{n}{p}.
\]
\end{proof}
We prove Theorem~\ref{mainlb}(b) by bounding $\max\{t,g(t)\}$.  
A lower bound on $\max\{t,g(t)\}$ is a value $t^*$ such that either $t \geq t^*$ or $g(t) \geq t^*$ for any $t$.  Therefore, we must prove $g(t) \geq t^*$ whenever $t \leq t^*$.  To prove this, we will use the fact that $g$ is a decreasing function and show that the value $t^*$ is such that $g(t) \geq t$ whenever $t \leq t^*$.  Indeed, this implies $g(t) \geq g(t^*) \geq t^*$ for all $t \leq t^*$.  

\begin{proof}[Proof of Theorem~\ref{mainlb}(b)]
We first state the value of a threshold $t^*$ such that $g(t)\ge t$ for all
$t\le t^*$ and then argue that it
implies Theorem~\ref{mainlb}(b).

\begin{lemma}\label{claim:t} Let 
$n,d,r,p$ be as in the statement of Theorem~\ref{mainlb}(b) and
$R = R(n,d,r)$ be as in Definition~\ref{defn-Dn}.
Define 
\begin{equation*}
\beta=\begin{cases}\gamma/2&\mbox{if $\gamma r>1$}\\
\gamma&\mbox{if $\gamma r\le 1$}.\end{cases}
\end{equation*}
Then, $g(t) \geq t$ whenever 
$1\leq t \leq t^*$ for \ \ 
$
	t^* = t^*(n,d,p,r) \ \deq \  \displaystyle\frac{\delta d^{\beta r -2r^2/R}}{64r^r R^{\gamma r}}.
$
\end{lemma}

By construction, $R$ is $\Theta(\log_d n)$
and the assumption that
$n \geq d^{r\log_2 d}$ implies that $d^{2r^2/R}$ is at most $c^{r}$ for
some constant $c$.
Plugging these values into the formula for $t^*$,
the two bounds of Theorem~\ref{mainlb}(b) follow immediately by combining
Lemma~\ref{yao} and 
Lemma~\ref{prop:g(t)-lower-bound}.

The lower bound proof for $\eqohbound$ is almost identical, but it uses a
variant definition of $g$ and Lemma~\ref{prop:g(t)-lower-bound} that replaces
the condition $(A_1,\ldots,A_p)\in A^{good}_\delta$ with the condition that the
conclusion of Proposition~\ref{exact-edge-density} holds, which is the 
only place that membership in $A^{good}_\delta$ is used in the proof
of Lemma~\ref{prop:g(t)-lower-bound}.
\end{proof}

It only remains to prove Lemma~\ref{claim:t}.

\begin{proof}[Proof of Lemma~\ref{claim:t}]   
Let $t\le t^*$ and $(A_1,\ldots,A_p)\in A^{good}_\delta$ be a tuple satisfying
$\sum_{i\in [p]}|A_i|=t\cdot 2^d$ and
$g(t)= \E_{S\sim \Dn}\left[ \max_i |A_i \cap S|\right]\cdot \frac{p}{n}$.
Since Lemma~\ref{lemma:max-to-ER} gives us a lower bound on $g(t)$ in terms of
$\sum_i |\ER(A_i)|$, our goal is to
lower bound the latter quantity, assuming $t~\leq~t^*$.  We rewrite this quantity in terms of the distribution $\mu$ 
over $[p]$ with $\mu(i) = \frac{|A_i|}{t2^d}$ as
%
\begin{equation}\label{fact2}
	\ \frac{1}{2^d} \sum_{i\in [p]}  |\ER(A_i)| 
	\ = \ \frac{1}{2^d} \sum_{i\in [p]}  \eR(A_i)|A_i| 
	\ = \ t \cdot \E_{i \sim\mu}[ \eR(A_i)].
\end{equation}
We apply Lemma~\ref{cor:path-lb} to lower bound the RHS in (\ref{fact2}).  Let $I\subseteq [p]$ be the indices $i
$ such that $A_i$ satisfies $\mathsf{e}_r(A_i) 
\geq 4\binom{R-r}{b+1} \binom{d}{r-b-1}$
for $b = \floor{\g r/2}$.  
Observe that with this choice of $b$, we have $\beta r\le b+1$ since
either $\beta r =\gamma r\le 1\le b+1$ or
$\beta r = \gamma r/2 \le \floor{\g r/2}+1=b+1$.
Using Jensen's 
inequality for $z \mapsto z^{R/r}$, we obtain
\begin{equation}\label{fact3}
	\E_{i \sim\mu}[ \eR(A_i)] 
	 \ \geq \ 
	\sum_{i \in I} \mu(i) \cdot\frac{\edr(A_i)^{R/r}}{\denom} 
	\ \geq\ 
	\left(\sum_{i \in I} \mu(i)\edr(A_i)\right)^{R/r}\cdot\frac{1}{\denom}. 
\end{equation}
%
Now, since $\mu(i)\edr(A_i)=\frac{|A_i| e_r(A_i)}{t 2^d}=\frac{|\Edr(A_i)|}{t 2^d}$, and $(A_1,\ldots,A_p)\in A^{good}_\delta$,
Proposition~\ref{exact-edge-density} implies that
\begin{equation*}
\sum_{i\in [p]} \mu(i)\edr(A_i)\ge \frac{\delta}{8t}\binom{d}{r}.
\end{equation*}
Therefore, to lower bound $\E_{i \sim\mu}[ \eR(A_i)]$,
it suffices to upper bound 
$\sum_{i \notin I} \mu(i)\mathsf{e}_r(A_i)$, which is at most
$4\binom{R-r}{b+1}\binom{d}{r-b-1}$ by definition.
Now, since $t\le t^*$, 
$\binom{d}{r}\ge (d/r)^r$ and $\beta r\le \min\{b+1,\gamma r\}$, we have
 \begin{equation*}
 \frac{\delta}{16t}\binom{d}{r}\ge \frac{\delta}{16t^*}\binom{d}{r}\ge
 4R^{\gamma r} d^{r-\beta r + 2r^2/R} > 4R^{b+1} d^{r-b-1}\ge 
 \sum_{i \notin I} \mu(i)\mathsf{e}_r(A_i).
 \end{equation*}
%
Therefore,
\begin{equation*}
\sum_{i \in I} \mu(i)\mathsf{e}_r(A_i)
\ge \frac{\delta}{16t}\binom{d}{r}\ge
4R^{\gamma r} d^{r-\beta r + 2r^2/R},
\end{equation*}
and from (\ref{fact3}) we have
\begin{equation}
\E_{i \sim\mu}[ \eR(A_i)]  \ \ge\ 
\left(4R^{\gamma r} d^{r-\beta r+2r^2/R}\right)^{R/r} \cdot\frac{1}{\denom} 
\ \ge \ \frac{R^{\gamma R} d^{R-(\beta R+bR/r)+2r}}{R!} 
\label{fact4}
\end{equation}
Now if $\gamma r\le 1$ then $b=0$ and $\beta=\gamma$ so $\beta R+bR/r=\gamma R$;
alternatively, if $\gamma r>1$, then $b\le \gamma r/2$ and $\beta=\gamma/2$
and we have $\beta R +bR/r\le \gamma R$.
Therefore, from (\ref{fact4}) we have
\begin{equation}
\E_{i \sim\mu}[ \eR(A_i)] 
\ \ge\  \frac{R^{\gamma R} d^{R-\gamma R+2r}}{R!}.\label{fact4.5}
\end{equation}
Using Lemma~\ref{lemma:max-to-ER} and plugging in  (\ref{fact2}) and
(\ref{fact4.5}) and the definition of $\gamma=\log_n p$ into the definition of 
$g(t)$, we have
\begin{equation*}
	g(t) \ = \
	\frac{\E_{S\sim \Dn}\left[ \max_i |A_i \cap S|\right]}{n^{1-\g}}
	\ \geq\ t\cdot \frac{R^{(r)}r!}{2d^{2r} \cdot n^{1-\g}}\cdot
            \frac{R^{\gamma R} d^{R-\gamma R+2r}}{R!}.
\end{equation*}
Since $R\ge 2r\ge 2$, we see that
$g(t)\ge t\cdot \frac{R^{\gamma R} d^{(1-\gamma)R}}{n^{1-\gamma}}$.
Since $n\le B(d,R)\le d^R/R!$, we derive $g(t)\ge t$ as required.
 \end{proof}
 



%% file: 6-conclusion.tex
\section{Discussion and Future Work}
We provided improved parallel algorithms for similarity joins  under Hamming distance.  We also proved communication lower bounds for one-round, local algorithms.  Qualitatively, we showed  that an overhead of $d^{\Theta(r)}$ is necessary and sufficient.   Although we stated our technical results for Hamming distance, we gave a template for upper and lower bounds in general edge-transitive similarity graphs. A main algorithmic theme running through our results was to perform upfront data replication so that the processors need only one round of communication.   This methodology may lead to improved algorithms for other distributed computation tasks.  

\paragraph{Local Computation.}   We optimized for the maximum
number of vertices assigned to any one processor and hence the maximum
difficulty of the local computation required;
however, our focus on communication leaves unanswered details about the actual
computation of the close pairs.  After the communication, any local algorithm
may be used to compare candidates and output pairs satisfying the distance threshold.  For example, processors could simply compare all received pairs.  Clearly, by enlarging the groups of points sent to each processor compared to~\cite{afrati-fuzzy,  afrati-anchor, afrati-dist} we lose the efficiency of being able to only check pairs of points within small subcubes or balls.  However, during the local computation stage, the processor may use a local partitioning scheme, possibly involving subcubes or balls, to reduce the overall number of comparisons.

Locality Sensitive Hashing (LSH) has been used successfully in practice for
range queries and nearest neighbor search~\cite{falconn, bahmani, HIM, lsh-practical}.
Recent work~\cite{coverlsh1, coverlsh2} exhibits an LSH-scheme for Hamming distance without false negatives. 

Overall, we believe that further optimizations are needed to implement an efficient all-pairs similarly search. In the $\ell_2$ metric, Aiger, 
Kaplan, and Sharir~\cite{sharir} design an algorithm for finding all close pairs using an  edge-covering consisting of  randomly shifted and rotated grids. Moreover, they demonstrate improvements over direct LSH approaches.    

Our approach also extends naturally to support dynamic and streaming data efficiently: only the processors responsible for a vector need to be notified for its addition or deletion.  Regarding randomness used for communication, work in the theoretical computer science community suggests bounded independence often suffices and improves performance~\cite{random2,  random3, random1}.

\paragraph{Open Questions.}
We point out five concrete future directions left open by our work.
\begin{enumerate}
\item Can we close the gap between our upper and lower bounds on $\ohbound$? Currently, there is a gap between the exponents of $\g r/2$ versus $r/2$.
\item Our randomized edge-covering needed a random construction to find a family of Hamming balls that cover all edges with high probability.  Is it possible to remove the error and decrease the communication by exhibiting an explicit construction? 
\item Our lower bound holds for the restricted model where processors must communicate sets of vertices.  Can we generalize our result and prove a lower bound on the maximum number of {\em bits} some processor must receive during the one round of communication? 
\item The protocol for arbitrary similarity graphs uses copies of the optimal edge-isoperimetric shape.  Can we determine this shape and analyze the resulting protocol for other metrics? Prime candidates include the $\ell_1$ distance and the edit distance over small alphabets. 
\item Does our edge-covering methodology lead to improvements in practice? Empirical research on similarity joins often employs a ``filter-then-verify'' strategy.  First, the algorithm or data structure coarsely prunes the set of candidates, often conservatively. Then, a brute force search reveals actual close pairs.  It is worth exploring filters based on edge-optimal sets. 
\end{enumerate}

%% file: 7-acknowl.tex
\section*{Acknowledgements}
We thank Sivaramakrishnan Natarajan Ramamoorthy for bringing Sidorenko's conjecture 
and related results 
to our attention.

%% file: A-appendix.tex
\section{Ball Ratios}
\begin{proposition}
\label{ballbound}
For $r< R\le (d+1)/2$, we have
\[\frac{B(d,R)}{B(d,R-r)}\le \frac{(d+1)d^{r-1}}{R^{(r)}}.\]
\end{proposition}

\begin{proof}
For $i\ge 0$ define
\begin{equation*}
b_i \deq\frac{\binom{d}{R-i}+\cdots+\binom{d}{R}}{\binom{d}{R-i}}.
\end{equation*}
Both $B(d,R)$ and $B(d,R-r)$ contain the common terms
$\binom{d}0+\cdots +\binom{d}{R-r}\ge \binom{d}{R-r}$ and so
\begin{equation*}
\frac{B(d,R)}{B(d,R-r)}\ \le\  b_r=\frac{\binom{d}{R-r}+\cdots+\binom{d}{R}}{\binom{d}{R-r}}.
\end{equation*}
Observe that by definition $b_0=1$ and for $i\ge 0$ we have the recurrence
\begin{align*}
b_{i+1}=\ 1+\frac{\binom{d}{R-i}}{\binom{d}{R-i-1}}\cdot b_i 
=\ 1+\frac{d-R+i+1}{R-i} \cdot b_i.
\end{align*}
We prove by induction that for $i\ge 1$ we have
$2\le b_i\le (d+1)d^{i-1}/R^{(i)}$.
For the base case, since $b_0=1$ we have
$b_1=1+\frac{d-R+2}{R-1} b_0 = (d+1)/R$ and since $R\le (d+1)/2$, we have $b_1\ge 2$.
For $i\ge 1$,
\begin{align*}
b_{i+1}&=1+\frac{d-R+i+1}{R-i}\cdot b_i\\
&= 1+\left(\frac{d}{R-i}-\frac{R-i-1}{R-i}\right)\cdot b_i\\
&\le 1+\frac{d}{R-i}\cdot b_i -\frac{R-i-1}{R-i}\cdot 2\qquad
\mbox{since $b_i\ge 2$ by the inductive hypothesis}\\
&\le \frac{d}{R-i}\cdot b_i. \qquad\mbox{since $R> r \geq i$.}
\end{align*}
Now applying the inductive hypothesis we have
\begin{equation*}
b_{i+1}\ 
\le \ \frac{d}{R-i}\cdot b_i
\ =\ \frac{d}{R-i}\cdot (d+1)d^{i-1}/R^{(i)}
\ = \ (d+1)d^i/R^{(i+1)}
\end{equation*}
as claimed.
\end{proof}
\section{Other Lower Bounds}\label{otherbounds}  We sketch two simple, but nontrivial, lower bounds for Hamming distance that achieve a different dependence on $|S|$ than the bounds we obtain using isoperimetric inequalities.  For the discussion, let $L$ denote the number of vertices a processor receives, not bits.  

For the first hard distribution, pick one random ball of radius $r/2$ and include each vector in this ball in $S$ with probability half.  A processor may output at most all pairs among vertices it receives. On the other hand,  all pairs  in the input $S$ are within distance $r$.  Thus, we need \[p\binom{L}{2} \geq  \binom{B(d,r/2)}{2}.\]  Since $|S| = \Theta(B(d,r/2))$, we can rewrite this as $L \geq \Omega(|S|/\sqrt{p}).$

For another hard distribution, consider picking $|S|/B(d,r/2)$ random balls each of radius $r/2$, then subsampling each vector with probably half.   Receiving $L$ vertices leads to at most $\binom{L}{2}$ close pairs. The input has $\Theta(|S|\cdot B(d,r/2))$ close pairs.   This gives 
\[ L \geq \Omega\left( \sqrt{B(d,r/2)} \cdot \sqrt{|S|/ p} \right). \]